\documentclass[11pt]{article}  

\usepackage{amsmath,amsthm,latexsym,amssymb,amsfonts,epsfig,braket,mathtools,yhmath}

\usepackage{graphicx}
\usepackage{tikz}
\usetikzlibrary{cd}
\usepackage{xcolor}
\usepackage{cleveref}
\usepackage{dsfont}
\usepackage{stmaryrd}

\usepackage[backend=biber, style=numeric-comp, maxcitenames=3, maxbibnames=99, doi=true, url=false, eprint=true,sorting=none]{biblatex}
\newtheorem{theorem}{Theorem}[section]
\theoremstyle{definition}\newtheorem{definition}{Definition}[section]
\newtheorem{corollary}{Corollary}[theorem]

\newcommand{\be}{\begin{equation}}
\newcommand{\ee}{\end{equation}}
\newcommand{\ba}{\begin{eqnarray}}
\newcommand{\ea}{\end{eqnarray}}

\title{{\sf Hamiltonian renormalisation: a Categorical Perspective}}
\author{
{\sf M. Rodriguez Zarate $^1$}\thanks{{\sf 
melissa.rodriguez@gravity.fau.de}}\\
\\
{$^1$\sf Inst. for Theor. Phys. III, FAU Erlangen -- N\"urnberg,}\\
{\sf Staudtstr. 7, 91058 Erlangen, Germany}\\
}
\date{{\small\sf \today}}

\begin{document}

\maketitle

{\sf

\begin{abstract}
We present a categorical formulation of the Hamiltonian renormalisation programme for quantum field theories, 
establishing a systematic bridge between functional and lattice renormalisation. 
To this end, we introduce two categories, \textbf{Seq} and \textbf{Func}, 
whose objects correspond to resolution spaces at different ultraviolet scales, 
and whose morphisms encode embeddings, projections, coarse-graining maps, and discrete derivatives.

Focusing on Dirichlet-type embeddings, we construct the corresponding subcategories $\textbf{Seq}_D$, $\textbf{Func}_D$ and prove that the embedding and its adjoint define functors between them. 

Furthermore we revisit and extend the analysis of the convergence rate to the fixed point for the couplings of the $U(1)^3$ model for $3+1$ Euclidean quantum gravity, analysing different combinations of Haar and Dirichlet embeddings.
\end{abstract}

\section{Introduction}

In quantum field theory (QFT), quantum fields are operator-valued distributions. When interactions are included, products of such distributions become ill-defined, leading to the appearance of the so called ultraviolet (UV) and infrared (IR) divergences. Renormalisation originates in QFT as a tool to control these singularities \cite{Weinbergrenormalisation}. In the constructive quantum field theory approach (CQFT) \cite{Haagbook,GlimmJaffebook} one tames the singularities by imposing temporal UV and IR cutoffs, discarding contributions from very high- and low-momenta. The dependence on these cutoffs is parametrised by the coupling constants of the theory. Then the regulators are removed by taking their appropiate limits: \textit{thermodynamic}
for the IR, \textit{continuum} for the UV. One expects that only a finite set of couplings —--known as relevant couplings--— survive after taking these limits. If this is not the case, the theory loses predictability and becomes ill-defined.

Renormalisation in QFT is usually treated perturbatively; gravity, however, is perturbatively non-renormalisable \cite{sagnotti,uvtwoloops}. Approaches to quantum gravity therefore address ultraviolet divergences using various background-independent renormalisation methods, most of which are formulated within the path-integral framework \cite{Dittrichsone,Steinhaus,Bahrone,Livine,Zapata,astrid}. In the case of canonical loop quantum gravity (LQG) \cite{Thiemannbook,Rovellibook,gambinipullin,Ashtekarbook} the theory is ultraviolet finite, but suffers from numerous quantisation ambiguities \cite{qsdiii,Perez}. The Hamiltonian renormalisation (HR) programme was introduced to precisely address these ambiguities \cite{renormalizationI,thiemannreview}. HR is inspired by the \textit{renormalisation group} \cite{Kadanoff,Wilsonren,Fisher,Wegner,Glazek}. Although originally developed with loop quantum gravity in mind, it is a theory--agnostic framework that can be applied to any canonical QFT, in both Euclidean and Lorentzian spacetimes. Moreover, HR is intrinsically non--perturbative and background--independent. The programme has been applied to free field theories \cite{renormalizationII,renormalizationIII,renormalizationIV,renormalizationV,renormalizationVI,renormalisationVII}, and recently, it has been extended to interacting QFTs \cite{renormalisationVIII,RenormalisationIX}.

From a mathematical standpoint, Hamiltonian renormalisation draws inspiration from constructive quantum field theory. In CQFT the basic objects are bounded quantum operators whose linear combinations and products are also bounded, forming an \textit{algebra of observables} $\mathcal{A}$. The physical content of the theory is captured by expectation values, represented mathematically by normalised positive linear functionals $\omega$ on $\mathcal{A}$. Each such functional corresponds uniquely to a Gelfand–Naimark–Segal (GNS) cyclic representation, which produces a Hilbert space $\mathcal{H}_\omega$ and a vacuum state $\Omega_\omega$ annihilated by the Hamiltonian $\hat{H}_\omega$ \cite{Haagbook}. The task is to renormalise either $\omega$ or its GNS data ($\mathcal{H}_\omega,\hat{H}_\omega, \Omega_\omega$).  
Starting from the GNS data at a UV resolution $M'$, one constructs an effective description at a coarser scale $M<M'$, by systematically removing the degrees of freedom between $M$ and $M'$. Iterating this procedure generates a flow of effective theories ---parametrised by the couplings--- whose fixed point is to be related to the continuum theory.

In the series of papers \cite{renormalizationII,renormalizationIII,renormalizationIV,renormalizationV,renormalizationVI,renormalisationVII}, 
the authors applied Hamiltonian renormalisation to free bosonic and fermionic field theories 
as a consistency check, 
demonstrating that the programme correctly reproduces the known exact continuum results. In \cite{renormalisationVIII,RenormalisationIX} the authors renormalised for the first time interacting QFTs; $P(\phi)_2$ polynomial self-interacting QFTs in two spacetime dimensions \cite{Simonone} and the $U(1)^3$ model for $3+1$ quantum gravity \cite{Smolinuone}. As in any renormalisation scheme, to construct the GNS data at UV resolution $M$, one has to make several choices since different discretisations, projections, representations, factor orderings etc. may lead to different continuum theories or affect 
the speed of convergence to the fixed point. 
In \cite{renormalisationVIII,RenormalisationIX} for instance, the authors compared three such choices: (i) \textit{Functional vs discrete renormalisation}, where fields are either projected onto Hilbert subspaces or discretised on a lattice;
(ii) \textit{Locality vs quasi-locality}; to project or discretise the fields Dirichlet (smooth, quasi-local) and Haar (discontinuous, local) kernels were chosen \cite{haar,dirichlet,Thiemannkernels}. (iii) \textit{Fock vs Narnhofer-Thirring (NT) representations} \cite{NTreps}.

Regarding points (i) and (ii), the corresponding choices are encoded within the HR framework through the definition of several linear maps ---such as embeddings, projections, inclusions, and discrete derivatives--- 
which specify the GNS data at different UV resolutions.
In this contribution, we focus precisely on these aspects and aim to address, among others, the following questions: 
Which discretisations and projections are most suitable? 
Can one consistently combine Haar and Dirichlet kernels? 
How does the rate of convergence towards the fixed point depend on the chosen kernel? 
What is the precise relationship between the projector-based and discretised formulations, and are they equivalent? 

Although these questions have already been raised in the HR series, here we will address them from a categorical point of view. As it will be shown, category theory provides a natural language to construct the UV resolution spaces and their associated maps, endowing HR with an elegant, axiomatic framework for analysing the various discretisation and projection choices. Beyond formal clarity, this perspective also enables the identification and classification of key structural properties of the discretisation and coarse--graining procedures, 
clarifying and providing a conceptual toolkit to systematically understand how different choices affect the corresponding renormalisation flow.

\vspace{3mm}

This work is organised as follows:

\vspace{3mm}

In Section 2, we introduce the Hamiltonian renormalisation programme, following closely \cite{renormalizationI,renormalisationVIII,renormalisationVII,Thiemannkernels}. 
In particular, we present both the functional and discretised formulations of the renormalisation procedure. 

Section 3 reformulates the construction of the UV resolution spaces in categorical terms, 
for both the projected (functional) and discretised cases. Furthermore we illustrate the aforementioned choices with the example of the $U(1)^3$ model for $3+1$ quantum gravity, where we both revisit and extend the analysis performed in \cite{RenormalisationIX}. 

In Section 4, we relate the functional and lattice renormalisation using functors. 

Finally in Section 5 we summarise and conclude.

\section{Hamiltonian renormalisation}
\label{hamiltonianrenormalisation}

More details on this section, in particular the relation to wavelet theory \cite{wavelet,waveletstwo} can be found in \cite{Thiemannkernels}.
We work on spacetimes diffeomorphic to $\mathbb{R}\times\sigma$. In a first step the 
spatial D-manifold $\sigma$ is compactified to a D-torus $T^D$, this effectively imposes periodic boundary conditions and corresponds to an IR cutoff which is fixed once and for all. Due to the compactification, the constructions 
that follow have to be done direction wise for each copy of $S^1$. On $X:=S^1$, understood 
as $[0,1)$ with endpoints identified, we consider 
the Hilbert space $L=L_2([0,1),\;dx)$ with elements from now on defined by capitalised letters $F ,G \in L$. Due to periodicity the elements of $L$ can be spanned by Fourier modes
\begin{align}
\label{a.1}
e_n(x):= e^{2\pi\;i\;n\;x},\; n\in\mathbb{Z},
\end{align}
which are orthonormal with respect to the inner product 
\begin{align}
   \label{a.2}
<F,G>_L:=\int_0^1\;dx\; \overline{F(x)}\; G(x). 
\end{align} 
The UV cutoffs must be labelled, with each label indicating the energy scale or resolution of the corresponding effective theory. To this end, we define the set $\mathbb{O} \subset \mathbb{N}$ of positive odd integers, which we use to index the UV resolutions. In order to relate lower and higher energy modes, we equip $\mathbb{O}$ with a partial order;
\begin{align}
\label{a.3}
M<M' \;\;\Leftrightarrow\;\; \frac{M'}{M}\in \mathbb{N},
\end{align}
The partial order is directed, that is, for each $M,M'\in \mathbb{O}$ we find 
$M^{\prime\prime}\in \mathbb{O}$ such that $M,M'<M^{\prime\prime}$ e.g. $M^{\prime\prime}=M M'$. We call the set $\mathbb{O}$, together with the partial order defined above, a partially ordered set (poset), denoted by $P_{\mathbb{O}}$. 
Each resolution $M\in P_{\mathbb{O}}$ restricts the real and momentum spaces to subspaces of finite resolution $M$. In the context of lattice renormalisation, this restriction would correspond to introducing a finite lattice in configuration space and its dual.
The subspaces are defined by
\begin{align}
\label{resolutionspaces}
\mathbb{N}_M=\{0,1,..,M-1\},\; 
\mathbb{Z}_M=\{-\frac{M-1}{2}, -\frac{M-1}{2}+1,..,\frac{M-1}{2}\},\;
X_M=\{x^M_m:=\frac{m}{M},\; m\in \mathbb{N}_M\},    \end{align}
Notice that $X_M$ satisfy $X_M\subset X_{M'} \;\;\Leftrightarrow\;\; M<M'$.
\subsection{Resolution spaces}
\label{resspaces}

Given the sets (\ref{resolutionspaces}), 
one may project down the elements in $L$ to the UV resolution spaces \(L_M \subset L\) defined by 
\be \label{basis}
L_M:={\sf span}(\{e_n,\; n\in \mathbb{Z}_M\}).
\ee
Its elements will be denoted by capital letters with an $M$-subscript, e.g.\ $F_M, G_M \in L_M$. In real space they correspond to compact, continuous functions, while in momentum space they are discretised with cutoff frequencies in $\mathbb{Z}_M$. 
On $L_M$ we use the same inner product as on $L$, hence the $e_n,\; n\in \mathbb{Z}_M$ 
provide an ONB for $L_M$. An alternative basis for $L_M$ can be constructed by means of suitable integral kernels. In the HR series, the authors have studied Haar and Dirichlet kernels. The motivation to introduce these bases is that, in contrast to the plane waves 
$e_n$, they have better locality properties; Dirichlet wavelets are spatially concentrated \textbf{around} $x=x^M_m$ while Haar bases are concentrated \textbf{at} $x=x^M_m$, furthermore both are real valued. In particular, Dirichlet bases are also smooth. This is a crucial feature because quantum field theory involves 
products of derivatives of the fields and derivatives of characteristic functions yield $\delta$ distributions. Further information regarding the two alternative bases will be discussed thoroughly sections 3 and 4.

Instead of the resolution spaces $L_M$ one may consider the space $l_M$ of square summable sequences $f_M=(f_{M,m})_{m\in \mathbb{N}_M}$
with $M\in P_\mathbb{O}$ and 
inner product 
\be \label{a.9}
<f_M,\;g_M>_{l_M}:=\frac{1}{M}\sum_{m\in\mathbb{N}_M}\; \overline{f_{M}(m)}\; g_{M}(m), 
\ee
where $f_{M,m}:=f_M(m)$. If we interpret $f_{M}(m)$ as $=F(\frac{m}{M})$ then (\ref{a.9}) can be viewed as the lattice approximant of 
$<F,G>_L$. Intuitively, working with the $L_M$ subspaces can be understood as projecting the elements of $L$ into the resolution subspaces $L_M$. On the other hand the sequences $f_M, g_M \in l_M$ correspond to lattice discretisations of $F_M, G_M \in L_M$ via the embeddings
\be \label{embedding}
I_M:\; l_M\to L_M;\;\;\; (I_M\cdot f_M)(x):=
<\chi^M_{\;\cdot}(x), f_M>_{l_M}=\frac{1}{M}\;\sum_{m\in \mathbb{N}_M}\; f_{M}(m)\; \chi^M_m(x).
\ee
Where the basis elements $\chi^M_m$ are (for the moment) left unspecified. The adjoint is defined by the requirement that 
\be \label{a.11}
<I_M^\ast\cdot F_M,\; g_M>_{l_M}= <F_M,\; I_M\cdot g_M>_{L_M},
\ee
and thus is given by 
\be \label{adjoint}
I_M^\ast:\; L_M\to l_M;\; (I_M^\ast\cdot F_M)(m)=<\chi^M_m, F_M>_{L_M}
\ee
Using (\ref{a.9}) one can find the identity map in $l_M$ and prove that $I_M$ is an isommetry
\be \label{a.13}
I_M^\ast \cdot I_M=1_{l_M}, \;\;<I_M .,I_M .>_{L_M}=<.,.>_{l_M},
\ee
which  also shows that $L_M, l_M$ are in 1-1 correspondence.
Likewise 
\be \label{a.13a}
P_M:=I_M \cdot I_M^\ast=1_{L_M}
\ee

\subsection{Functional renormalisation}
In this subsection we discuss how to implement 
the Hamiltonian renormalisation programme in the continuum, without resorting to an explicit lattice discretisation.  This is achieved by projecting the smeared fields into the resolution spaces $L_M\subset L$. In order to do so, we consider $I_M$ also as a map $I_M:\; l_M\to L$ with image $L_M\subset L$. The map $I_M^\ast: L\to l_M$ has the same formula of the previous subsection with $F\in L$, furthermore  $P_M: L\to L_M$ can now be understood as an orthogonal 
projection 
\be \label{a.14}
P_M\cdot P_M=P_M,\; P_M^\ast=P_M.
\ee
One can then project the functions $F \in L$ into $L_M$ via
\be \label{a.15}
F_M(x):=(P_M\cdot F)(x)=\int_0^1\; dy\; P_M(x,y)\; F(y).
\ee

Given a functional $H[\Pi,\Phi]$ of the continuum fields we define its projection at resolution M 
by 
\be \label{a.16b}
H_M[\Pi_M,\Phi_M]=H[\Pi_M,\Phi_M].
\ee
To project continuous phase space functionals (such as the Hamiltonian) into the resolution spaces $L_M$ one projects the canonical conjugate pairs, that is
\begin{align}
    H_M[\Pi_M,\Phi_M]:=H[\Pi_M,\Phi_M].
\end{align}
This projections have to be carried to the quantisation. In order to do so, we briefly remind the reader the quantisation from a CQFT perspective. The (Weyl) algebra of observables $\mathfrak{A}$ correspond to bounded operators in L generated by the Weyl elements 
\be \label{a.18}
W[F]=e^{-i<F,\Phi>_L},\; W[G]=e^{-i<G,\Pi>_L}
\ee 
for real valued $F,G\in L$. The representation and commutation relations are encoded in the Weyl relations
\ba \label{a.19}
&& W[G]\; W[F]\; W[-G]=e^{-i<G,F>_L}\;W[F],\; 
W[F]\;W[F']=W[F+F'], \; W[G]\;W[G']=W[G+G']
\nonumber\\
&& W[0]=1_{\mathfrak{A}},\;
W[F]^\ast=W[-F],\; W[G]^\ast=W[-G] 
\ea
Via the GNS construction one can identify a one--to--one correspondence between the states (positive, normalised linear functionals) $\omega$ on $\mathfrak{A}$ and a GNS data given by a Hilbert space $\mathcal{H}_\omega$ and a cyclic representation $\rho, \Omega$ \cite{Haagbook,GN,S}. 
The correspondence is given by 
\be \label{a.20}
\omega(A)=<\Omega, \; \rho(A)\Omega>_{{\cal H}}. 
\ee
For each $M \in P_\mathbb{O}$ the Weyl algebra as well as their relations are replaced by their projected versions
\begin{gather} \label{a.21}
W_M[F_M]=e^{-i<F_M,\Phi_M>_{L_M}},\;\;
W_M[G_M]=e^{-i<G_M,\Pi_M>_{L_M}}, \\
W_M[G_M]\; W_M[F_M]\; W_M[-G_M]=e^{-i<G_M,F_M>_{L_M}}\;W_M[F_M],\; 
W_M[F_M]\; W_M[F'_M]=W_M[F_M+F'_M],\; \nonumber
\\
W_M[G_M]\; W_M[G'_M]=W_M[G_M+G'_M],\;
W_M[0]=1_{\mathfrak{A}_M},\;W_M[F_M]^\ast=W_M[-F_M],\; W_M[G_M]^\ast=W_M[-G_M]
\end{gather}
The GNS construction can be carried out at each resolution $M$, thus a state $\omega_M$ on $\mathfrak{A}_M$ is identified with the GNS data $(\rho_M, {\cal H}_M, \Omega_M)$ where the representation is understood to be densely defined on the subspace ${\cal D}_M=
\mathfrak{A}_M \Omega_M$. 
Due to the poset structure the algebras are nested, for instance, given $M<M'\in P_{\mathbb{O}}$ then $\mathfrak{A}_M\subset \mathfrak{A}_M'\subset \mathfrak{A}$. This nesting property follows directly from the projectors and implies the  identities
\be \label{a.23}
W_{M'}[F_M]=W_M[F_M],\;\;W_M[F_M]=W[F_M]
\ee
To initialise the renormalisation programme one needs to provide 
an Ansatz of a family of projected theories
$(\omega^{(0)}_M,\rho^{(0)}_M(H_M,c^{(0)}_M))_{M\in \mathbb{O}}$ at some effective resolution $M$.
Theories corresponding to different resolutions are connected through a renormalisation flow equation, which recursively generates a sequence $(\omega^{(r)}_M,\rho^{(r)}_M(H_M,c^{(r)}_M))_{M\in \mathbb{O}}$ here labelled by $(r)$. The elements of the sequence are plugged into the renormalisation flow equation iteratively. As \(r\) increases, the sequence may converge to an accumulation point
\((\omega^{(*)}_M,\, \rho^{(*)}_M(H_M, c^{(*)}_M))_{M \in \mathbb{O}}\),
known as the \emph{fixed point}, where the theory becomes \emph{scale invariant}. In the HR programme, both the algebraic state $\omega_M$ and the matrix elements of the Hamiltonian operator $H_M$ (or quadratic form) are renormalised through the renormalisation flow equations
\ba\label{a.26}
&& \omega^{(r+1)}_M(A_M):=\omega^{(n)}_{M'(M)}(A_M),\;\; 
<\rho_M^{(r+1)}(A_M)\Omega^{(r+1)}_M,\;\rho^{(r+1)}_M(H_M)\; \rho^{(r+1)}_M(A'_M)\Omega^{(r+1)}_M>_{{\cal H}^{(r+1)}_M}
\nonumber\\
&=& <\rho^{(r)}_{M'}(A_M)\Omega^{(r)}_{M'},\;\rho^{(r)}_{M'}(H_{M'})\; \rho^{(r)}_{M'}(A'_M)\Omega^{(r)}_{M'}>_{{\cal H}^{(r)}_{M'}}
\ea
where $M':\mathbb{O}\to \mathbb{O}$ is a fixed map with the property that 
$M'(M)>M,\;M'(M)\not=M$. 

\subsection{Lattice renormalisation}

In the discrete formulation, the renormalisation procedure is implemented on a spatial lattice, 
where the continuum fields are replaced by finitely many degrees of freedom associated with the lattice sites and thus is closer in spirit to condensed matter renormalisation.

The formulation follows naturally by exploiting the relation between the continuum $L_M$ spaces and the discrete $l_M$ ones. This is achieved through the embedding map $I_M$ and its adjoint $\overline{I}_M$. The Weyl elements are then related via the identities 
$w_M[f_M]=W_M[F_M],\; f_M=I_M^\ast\cdot F_M$. The $w_M[f_M], \; w_{M'}[f_M']$ at resolution $M, \; M'$ respectively can be related via the \textit{coarse graining} map $i_{M M'}:= \overline{I}_MI_{M'};\; l_M\to l_{M'}$ such that
$w_{M'}[I_{M M'}\cdot f_M]=w_M[f_M]$. This map obeys $i_{M_2 M_3}\cdot i_{M_1 M_2}=i_{M_1 M_3}$ for 
$M_1<M_2<M_3$ because the image of $I_M$ is $L_M$ which is a subspace of $L_{M'}$ thus
$i_{M_2 M_3}\cdot i_{M_1 M_2}=I_{M_3}^\ast\cdot P_{M_2}\cdot  I_{M_1}=I_{M_3}^\ast\cdot I_{M_1}$. 
Then  $W_{M'}[F_M]=w_{M'}[I_{M'}^\ast F_M]=w_{M'}[I_{M'}^\ast\cdot I_M\cdot f_M]
=w_{M'}[i_{M M'}\cdot f_M]$ indeed. For the same reason $W_{M'}[F_M]=W_M[F_M]$ as 
$L_M$ is embedded in $L_{M'}$ by the identity map. The renormalization flow in terms of Weyl elements $w_M[f_M]$ and the coarse graining map $I_{M,M'}$ takes the form
\begin{align}
\label{rflow}
\omega_M^{(r+1)}(w_M[f_M])&:=\omega^{(r)}_{M'}(w_{M'}[i_{M,M'}f_M']) \nonumber \\ 
\braket{w_M[f'_M]\Omega_{M}^{(r+1)},H^{(r+1)}_{M}\;w_M[f_M]\Omega_{M}^{(r+1)}}_{\mathcal{H}^{(r+1)}_M}&:=\braket{w_{M'}[r_{M,M'}f'_M]\Omega_{M'}^{(r)},H^{(r)}_{M'}\;w_{M'}[I_{M,M'}f_M]\Omega_{M'}^{(r)}}_{\mathcal{H}^{(r)}_{M'}}.
\end{align}
\section{Categorical framework for resolution spaces}
In this section, we reformulate the structure of the UV resolution spaces using the language of category theory, we do so by introducing two closely related categories: the \textbf{Seq} category, whose objects are the spaces of square integrable sequences $l_M$ and the \textbf{Func} category, whose objects are the $L_M$ spaces. Throughout this section the embeddings $I_M$ are defined as in \cref{resspaces} and follow the properties depicted therein. We start by defining the Dirichlet and Haar embeddings.

\begin{definition}[Dirichlet embedding]
\label{demb}
Let $P_\mathbb{O}$ be the poset of resolution scales defined in \cref{resspaces} and $l_M$ be the space of square summable sequences at resolution $M\in P_{\mathbb{O}}$. The Dirichlet embedding $I_M^D:l_M\rightarrow L_M^D$ follows \cref{embedding} where \(\chi^{D, M}_{m}(x)\) corresponds to the Dirichlet kernel 
\begin{align}
\chi^{D,M}_{m}(x):=\sum_{n\in \mathbb{Z}_M}\; e_n(x-\frac{m}{M})=\frac{sin[M \pi (x-\frac{m}{M})]}{sin [\pi (x-\frac{m}{M})]},
\end{align}
and $L_M^D\subset L$ is the subspace of $L$ spanned by the Dirichlet basis $\{\chi^{D,M}_{m}\}_{m\in \mathbb{N}_M}$.
\end{definition}
The adjoint $\overline{I}_M^D:L^D_M\rightarrow l_M$ follows from \cref{adjoint}. The functions $\chi^{D,M}_{m}$ are orthogonal but not orthonormal \(
<\chi^M_{D,m},\;\chi^M_{D,m'}>_{L_M^D}=M\; \delta_{m,m'}\).

\begin{definition}[Haar embedding]
\label{hemb}
Let $P_\mathbb{O}$ be the poset of resolution scales defined in \cref{resspaces} and $l_M$ be the space of square summable sequences at resolution $M\in P_{\mathbb{O}}$. The Haar embedding $I_M^H:l_M\rightarrow L_M^H$ follows \cref{embedding} where \(\chi^{H,M}_{m}\) corresponds to the Haar kernel 
\begin{align}
\chi^{M}_{H,m}(x):=M\chi_{[\frac{m}{M},\frac{m+1}{M})}(x),
\end{align}
where $\chi$ is the indicator function and $L_M^H\subset L$ is the subspace of $L$ spanned by the Haar basis $\{\chi^{H,M}_{m}\}_{m\in \mathbb{N}_M}$.
\end{definition}
Once more, $\overline{I}_M^H:L^H_M\rightarrow l_M$ follows from \cref{adjoint}. 
The functions $\chi^{H,M}_{m}$ form an orthogonal basis for $L_M^H$ with inner product \(
<\chi^M_{H,m},\;\chi^M_{H,m'}>_{L_M^H}=M\; \delta_{m,m'}\). 
Both $L^D_M, \; L^H_M$ are subspaces of $L$.
However, in general, $L^D_M \neq L^H_M$. This is due to the fact that $L^D_M$ is the span of Dirichlet kernels (bandlimited trigonometric polynomials of degree $2\leq M/2$) while $L^H_M$ is the span of Haar boxes (piecewise constant on the
$M$ intervals). 
Since both $L^D_M, \; L^H_M \subset L$, we may relate them via a change of basis transformation. 
For instance, 
we can write any element of one basis as a linear combination of the other
\begin{align}
    \chi^{D,M}_{m}(x)&=\sum_{m' \in \mathbb{N}_M} A_{m',m}\;\chi^{H,M}_{m'}(x), \;\;\; A_{m',m}:=\frac{1}{M}\braket{\chi^{H,M}_{m'},\chi^{D,M}_{m}}_L, \label{dtoh}  \\
    \chi^{H,M}_{m'}(x)&=\sum_{m \in \mathbb{N}_M} B_{m,m'}\;\chi^{D,M}_{m}(x),\;\;\;B_{m,m'}:=\frac{1}{M}\braket{\chi^{D,M}_{m},\chi^{H,M}_{m'}}_L. \label{htod}
\end{align}
The bases are normalised such that the transformation is unitary;
since the elements in both bases are real-valued, the inner products in $A_{m',m}$ and $B_{m,m'}$ are equal. 
We compute
\begin{align}
\label{changeofbasis}
\braket{\chi^{H,M}_{m},\chi^{D,M}_{m'}}_L &=\sum_{n\in \mathbb{Z}_M}\left(\int_L \chi_{[\frac{m}{M},\frac{m+1}{M})}(x)\;e^{2\pi i nx}dx\right)e^{-2\pi i nm'/M}=\sum_{n\in \mathbb{Z}_M}\left(\int_\frac{m}{M}^{\frac{m+1}{M}} e^{2\pi i nx}dx\right)e^{-2\pi i nm'/M}.
\end{align}
For $n=0$, the integral in \cref{changeofbasis} gives $1/M$. For $n\neq 0$ we get
\begin{align}
    \int_\frac{m}{M}^{\frac{m+1}{M}} e^{2\pi i nx}dx= \frac{e^{2\pi i(m+1)/M}-e^{2\pi inm/M}}{2\pi in}=e^{2\pi inm/M}\frac{(e^{2\pi in/M}-1)}{2\pi in}=\frac{e^{\pi in/M}}{\pi n}sin(\pi n/M).
\end{align}
Raplacing in \cref{changeofbasis}, the inner product takes the form
\begin{align}
\braket{\chi^{H,M}_{m},\chi^{D,M}_{m'}}_L =\frac{1}{M}+2\sum_{n=1}^{\frac{M-1}{2}}cos[\frac{2\pi n}{M}(m-m'+1/2)]\frac{sin(\pi n/M)}{\pi n} 
\end{align}
Due to the linearity of the Hilbert space, \cref{dtoh,htod} can be viewed as linear isomorphisms:
\begin{definition}[Change of basis]
Let $M \in P_\mathbb{O}$, furthermore, let $L^D_M$ and $L^H_M$ denote the subspaces spanned by the Dirichlet and Haar bases, respectively. 
The change of basis from Dirichlet to Haar is given by the linear isomorphism $I_{M}^{D H}:L^D_M\rightarrow L^H_M$, whose action is defined in \cref{dtoh}. 
Conversely, 
the change of basis from Haar to Dirichlet is the linear isomorphism  $I_{M}^{H D}:L^H_M\rightarrow L^D_M$ whose action is given by \cref{htod}.
\end{definition}
Derivative operators are also discretised on the lattice. 
Here we follow two possible discretisations, related to the Haar and Dirichlet kernels:
\begin{definition}[Dirichlet derivative]
\label{partiald}
Let 
\begin{align}
    \partial^D: L^D_M \to L^D_M\;\;;\;\;\partial:=\partial\big|_{L^D_M},
\end{align}
denote the usual derivative operator on $L$ restricted to $L_M^D\subset L$.
The \emph{Dirichlet (discrete) derivative} is the operator 
\begin{align}
\partial_{M}^D := \overline{I}_{M}^D \cdot\partial^D\cdot I_{M}^D : \; l_M \longrightarrow l_M; \;\;\; f_M(m)\mapsto \sum_{m' \in \mathbb{Z}_M} \sum_{n \in \hat{\mathbb{Z}}_M}
(2\pi i n)\, e^M_{n}(m-m')\, f_M(m'). 
\end{align}
\end{definition}
The Haar derivative is more subtle: since the pointwise derivative of Haar functions is not in $L^2$, we define the Haar derivative in a \emph{distributional sense} by 
\begin{definition}[Haar discrete derivative]
\label{partialh}
Let $L_M^H \subset L$ be the finite-dimensional subspace spanned by the Haar basis $\chi_m^{H,M}$. Furthermore, let $\partial^H : L_M^H \longrightarrow \mathcal{D}'(L)$ be the distributional derivative whose action on the Haar basis is
\begin{align}
\partial \chi_m^{H,M} = \delta_{(m+1)/M} - \delta_{m/M}. 
\end{align}
The \emph{Haar (discrete) derivative} is defined as the map 
\begin{align}
\partial_M^H :=\overline{I}^H_M\cdot \partial^H \cdot I^H_M: l_M \longrightarrow l_M\;;\;\;
f_M(m) \;\mapsto\; \frac{1}{M}[f_M(m+1) - f_M(m)] 
\end{align}
\end{definition}
The distinction between the definitions arises because on the one hand, Dirichlet wavelets are smooth and thus allow well-defined derivatives at every resolution, 
so that the corresponding derivative operator $\partial^D$ converges to the continuum derivative as $M \to \infty$. 
In contrast, Haar functions are piecewise constant, so their derivatives $\partial^H$ are distributions supported only at discontinuities. 
As a result, $\partial^H$ does not admit a  pointwise continuum limit, ---rather a distributional one--- 
and any attempt to construct the continuum derivative from Haar modes necessarily yields to Dirac deltas. 
This is why, in the continuum formulation of Hamiltonian renormalisation, one works with $\partial^D$ rather than $\partial^H$.

In order to relate sequences at different resolutions one needs to construct the coarse graining maps. In \cite{renormalisationVIII,RenormalisationIX} the authors used only Dirichlet coarse grainings maps, while in earlier works \cite{renormalizationII,renormalizationIII}, Haar coarse grainings were used. In this section we consider only the Dirichlet ones. The generalisation follows naturally.
\begin{definition}[Coarse graining]
Let $M\in P_\mathbb{O}$ and $l_M, l_{M'}$ be the spaces of square summable sequences at resolution $M$ and $M'$ with $M<M'$ respectively. The \emph{Dirichlet coarse graining} map is defined by
\begin{align}
\label{coarsegraining}
I_{MM'}:=\overline{I}^D_{M'}\cdot I^D_M:l_M\rightarrow l_{M'};\;\;\; f_M(m)\mapsto \sum_{m' \in \mathbb{N}_M}   \sum_{n\in \mathbb{Z}_{M}} e_n^M(m/3-m')f_M(m').
\end{align}
\end{definition}
Now we come to one of the main results of this contribution. The resolution spaces $\{L^X_M\}_{M\in P_{\mathbb{O}}}\subset L$ together with the maps defined above can be wrapped together into a category, the formal statement is give by the following theorem:
\begin{theorem}
\label{Funccat}
Let \(P_{\mathbb O}\) be the poset of resolution scales defined in \cref{a.3}.
For each \(M\in P_{\mathbb O}\), let \(L_M^X\subset L:=L_2([0,1),dx)\) denote the finite-dimensional subspace spanned by the basis \(\{\chi^{X,M}_m\}_{m\in \mathbb{N}_M}\)
with \(X\in\{D,H\}\) (Dirichlet/Haar) and inner product \eqref{a.9}. For every $M< M'$ let  
\begin{align}
I_{M,M'}^X:L^X_M\hookrightarrow L^X_{M'}\;; \; I_{M,M'}:=P_{M'}P_M\big|_{L^X_M}
\end{align}
 be a canonical inclusion map where $P_M,\;P_{M'}$ are the orthogonal projections defined in \cref{a.14}. 
Let moreover $I_M^{DH}:L_M^D\to L_M^H$, $I_M^{HD}:L_M^H\to L_M^D$ be the change of basis isomorphisms defined in \cref{changeofbasis}, and $\partial$ be the usual derivative operator on $L_2$ spaces.
Then the collection whose objects are $\{L_M^X\}_{M,X}$ for $M\in P_{\mathbb{O}}$, $X\in \{D,H\}$ and whose morphisms are the maps
\begin{align}
  \label{arrowsfunc}
    \big\{\mathrm{Id}_{L_M^X}:=\mathrm{Id}_L\big|_{L^X_M},\; I^X_{M,M'},\; I_M^{DH},\; I_M^{HD},\; \partial^D\big\}
\end{align}
is a category of functions \(\mathbf{Func}\), where the composition is defined as the usual composition of linear maps. 
\end{theorem}
\begin{proof}
We start by realising that the maps $I^X_{M,M'}$ are indeed an inclusion.  The projectors are explicitly given by
\begin{align}
    \label{projection}
        (P^X_{M}F)(x)=\sum_{m\in \mathbb{N}_M}\chi^{X,M}_m(x)\braket{\chi^{X,M}_m,F}_L,
    \end{align}
    thus $P^X_{M'}P^X_M\big|_{L^X_M}=P^X_M$. Furthermore, the coarse basis elements $\chi^{X,M}_{m}$ are projected into the finer ones $\chi^{X,M}_{m'}$
    \begin{align}
    \label{partition}
        (P_{M'}\;\chi^{X,M}_{m})(x)=\sum_{l=0}^{k-1}\chi^{X,M'}_{km+l}(x), \;\; \text{for}\;\; k=\frac{M'}{M}.
    \end{align}
Composition of inclusion maps follows from \cref{projection,partition}
\begin{align}
\label{compositioninclusion}
I^X_{M',M''}I^X_{M,M'}=P^X_{M''}P^X_{M'}\;P^X_{M'}P^X_{M}=P^X_{M''}P^X_{M}=I^{X}_{M,M''}.
\end{align}
The change of basis maps $I^{D,H}_M,\; I^{H,D}_M$ are, by construction, linear isomorphisms. The identity map \(\mathrm{Id}_{L^X_M}\) is just the restriction of the identity on $L$ to $L^X_M$ and thus also a linear isomorphism. Regarding the derivative $\partial^D$, by construction it is the usual derivative of $L$ spaces restricted to $L^D_M$, furthermore it preserves the $L^D_M$ spaces. This follows from the fact that
\begin{align}
\label{commutation}
   [\partial^D, P_M^D]=0 \;\;\Longrightarrow \;P^D_M(\partial^DF)=\partial^DP^D_MF=\partial^DF,\;\;\; \mathrm{for} \;F \in L^D_M\; \mathrm{and} \;\forall \; M\in P_{\mathbb{O}},
\end{align}
which is computed directly from \cref{partiald} and \cref{projection}, hence $\partial^D$ is an endomorphism. 
Henceforth, all morphisms in \cref{arrowsfunc} are linear maps on $L^X_M \subset L$ and thus associativity and closure under composition follows. 
\end{proof}
Severall comments are in order:
\begin{enumerate}
    \item For $P_\mathbb{O}$, the inclusion maps provides us with a "nesting" of square integrable subspaces. For instance, given $M< M'< M''$ in $P_{\mathbb{O}}$ one has $L_M^X\subset L^X_{M'}\subset L^X_{M''}$.   
    \item The composition rule \cref{compositioninclusion} provides us with a cylindrical consistent structure. This implies that the hierarchy of subspaces \(\{L^X_M\}_{M\in P_{\mathbb O}}\) forms a 
\emph{direct system} of Hilbert spaces with compatible embeddings.  
In HR the direct system structure is crucial as it provides a natural construction of the continuum Hilbert space via inductive limits.

\item We do not consider $\partial^H$ in the category because it is not a bounded operator thus it is not a morphism in \textbf{Func}. 

\item The fact that $\partial^D$ preserves $L^D_M$ has important consequences in the renormalisation. This is stated in the following corollary
\end{enumerate} 

\begin{corollary}
    Let $\partial^D$ be given in \cref{partiald} and let \cref{Funccat} hold. Then $\partial^D$ is compatible with $I^D_{M,M'}$.
\end{corollary}
\begin{proof}
    Using \cref{commutation} one gets
    \begin{align}
    \label{partialcilindrical}
        \partial^D I^D_{M,M'}=\partial^D(P_M'P_M)=P_{M'}\partial^DP_M=P_{M'}P_M\partial^D=I^D_{M,M'}\partial^D,
    \end{align}
    for any $M,M' \in P_{\mathbb{O}}$ such that $M<M'$.
\end{proof}

Physically, equation~\eqref{partialcilindrical} expresses the cylindrical consistency of the discrete Dirichlet derivative with respect to the inclusion maps \(I^D_{M,M'}\). This guarantees that the discretised derivative operator 
\(\partial^D\) does not introduce spurious scale–dependent artefacts. \\

Similarly to \cref{Funccat}, the resolution spaces $\{l_M\}_{M\in P_{\mathbb{O}}}$ and their associated maps can be unified in categorical terms by means of the following theorem:

\begin{theorem}
Let \(P_{\mathbb O}\) be the poset of resolution scales defined in \cref{a.3}. For each $M\in P_{\mathbb O}$, let $l_M$ denote the space of square summable sequences at resolution $M$ with inner product given in \cref{a.9}, and let $I_M^D:l_M\to L_M^D\subset L$ be the Dirichlet embedding with adjoint $\overline I_M^D:L_M^D\to l_M$.  
Then the collection of objects $\{l_M\}_{M\in P_{\mathbb O}}$ together with the maps
\begin{align}
\{\partial_M^X,\;i_{MM'}:=\overline I_{M'}^D I_M^D,\;\mathrm{Id}_{l_M}^X:=\overline{I}^X_MI_M^X\},
\end{align}
for $X=\{D,H\}$, defines a category of sequences \textbf{Seq}, where the composition is the ordinary composition of linear maps.  \end{theorem}

\begin{proof}
For the composition of Dirichlet derivatives, we use \cref{commutation} to prove
    \begin{gather}
      \partial_M^D \circ\partial_M^D= \overline{I}^D_M \; \partial^D P^D_M \; \partial^D \; I_M^D=\overline{I}^D_M \;(\partial^D)^2\; I^D_M,
    \end{gather}
Moreover, from \cref{Funccat} we know the operator $\partial^D$ is closed under composition, thus $(\partial^D)^2:L^D_M\longrightarrow L^D_M$ from which closeness under composition of Dirichlet derivatives follow. 
\vspace{1mm}
Haar derivatives are gain more subtle since \cref{commutation} no longer holds. Instead, we use directly \cref{partialh}  
\begin{align}
   (\partial^H_M)^2 f_M(m)=\frac{1}{M^2}\big[f_M(m+2)-2f_M(m+1)+f_M(m)\big]. 
\end{align}
The right-hand side is clearly a sequence indexed by $m\in\mathbb Z_M$, so it lies in $l_M$. Linearity is immediate from the formula, hence $(\partial^H_M)^2$ is a linear endomorphism of $l_M$. For the composition of embedding maps, let $M<M'<M'' \in P_{\mathbb{O}}$, then 
\begin{align}
\label{cylindrical}
i^D_{M_2,M_3}i^D_{M_1,M_2}=\overline{I}^D_{M_3}\;I^D_{M_2}\;\overline{I}^D_{M_2}\;I^D_{M_1}=\overline{I}^D_{M_3}\;P^D_{M_2}\;I^D_{M_1}=\overline{I}^D_{M_3}\;I^D_{M_1}=i^D_{M_1,M_3},
\end{align}
associativity follows directly from the linearity of the maps, equivalently, one can prove it using \cref{cylindrical} 
\begin{align}
i_{M_3,M_4}(i_{M_2,M_3}i_{M_1,M_2})=i_{M_3,M_4}i_{M_1,M_3}=i_{M_1,M_4},
\end{align}
and similarly for the left associativity.
Regarding the identity, it is easy to see that the maps $\mathrm{Id}_{l_M}^X$ are indeed the identity in $l_M$, for instance
\begin{align}
    (\overline{I}^X_MI^X_M\;f_M)(x)=\frac{1}{M}\sum_{m'\in \mathbb{N}_M}f_M(m')\;\braket{\chi^{X,M}_m,\; \chi^{X,M}_{m'}}_{L^{X}_M}=f_M(m),
\end{align}
furthermore $\mathrm{Id}_{l_M}^X=\overline{\mathrm{Id}}_{l_M}^X$ and $\mathrm{Id}_{l_M}^X\;\mathrm{Id}_{l_M}^X=\mathrm{Id}_{l_M}^X$. Now we prove the left and right composition between the identity and the other morphisms in \textbf{Seq}
\begin{align}
\partial^X_M\;\mathrm{Id}_{l_M}^X&=\overline{I}^X_M \;\partial^XI^X_M\;\overline{I}^X_M \;I^X_M=\overline{I}^X_M \;\partial^X P_M^X\;I^X_M=\partial^X_M,\\
\mathrm{Id}_{l_M}^D\;\partial^D_M&=\overline{I}^D_M \;I^D_M\; \overline{I}^D_M \;\partial^DI^D_M=\overline{I}^D_M \;P^D_M\;\partial^DI^D_M=\partial^D_M, 
    \end{align}

For $\mathrm{Id}_{l_M}^H\;\partial^H_M$ once again we cannot make use of the commutativity of the projector and $\partial^ H$, instead we use the properties of the orthogonal projection and \cref{a.11}. For all $f_M,g_M\in l_M$ and $M\in P_{\mathbb{O}}$ we have
\begin{align}
\label{identityone}
\braket{\mathrm{Id}_{l_M}^H\partial^H_M\;f_M,g_M}_{l_M}=\braket{\partial^H_M\;f_M,\;\mathrm{Id}_{l_M}^Hg_M}_{l_M}=\braket{\partial^H_M\;f_M,\;g_M}_{l_M}.
\end{align}
For the closeness under composition of $\mathrm{Id}^{D}_M\partial^H_M$, $\mathrm{Id}^{H}_M\partial^D_M$ and $\mathrm{Id}^{X}_Mi^{D}_{M,M'}$ the same steps as in \cref{identityone} are performed. The related cases where the identity is to the right follow directly since $\overline{I}^X_MI^X_Mf_M(m)=f_M(m)$. Associativity under composition is direct since all the maps are linear.  
\end{proof}
\subsection{Example: The $U(1)^3$ model for 3+1 quantum gravity}
To illustrate how different choices of embeddings, discretisations, and coarse-graining maps affect the renormalisation flow and the rate of convergence towards the fixed point, 
we will briefly revisit and extend the computations performed in \cite{RenormalisationIX}.

The $U(1)^3$ model can be regarded as a week (Newton constant) limit of Euclidean signature general relativity \cite{Smolinuone}. The details of its quantisation can be found in \cite{thiemanexact}. We refer to \cite{RenormalisationIX,Thiemannfock} for all the details regarding renormalisation. We consider spacetimes of the form $\mathbb{R} \times \sigma$, where $\sigma$ is a three-manifold compactified with periodic boundary conditions. In this setting, we restrict our attention to the polynomial Hamiltonian constraint
\begin{align}
\label{hamiltonianconstraint}
H[N]:= \int_\sigma\; d^3x\; A_a^j\; H^a_j(N,E),\; H^a_j(N,E):=\epsilon_{jkl}\;(N\;E^{[a}_k E^{b]}_l)_{,b},
\end{align}
where $(A_a^j, E^a_j)$ are the conjugate pairs of variables coordinatising the real phase space, and $N$ denotes the scalar lapse function, which plays the role of the to--be--renormalised coupling parameter. In~\cite{RenormalisationIX}, the coupling parameters were embedded using both Haar and Dirichlet bases. For the Haar kernels we have
\begin{align} \label{haarcouplings}
    N_{M,b}(m_1,m_2,m_3)&:=N_{,b}\left(\frac{m_1}{M}\right)\delta_{m_1,m_2}\delta_{m_1,m_3}, \;\;\;
N_M(m_1,m_2,m_3):=N\left(\frac{m_1}{M}\right)\delta_{m_1,m_2}\delta_{m_1,m_3},
\end{align}
while the Dirichlet couplings are 
\begin{align}
   N_{M,b}(m_1,m_2,m_3)&:= \braket{N_{,b},\chi^{D,M}_{m_1}\chi^{D,M}_{m_2}\chi^{D,M}_{m_3}}, \;\;\;  N_{M}(m_1,m_2,m_3):= \braket{N,\chi^{D,M}_{m_1}\chi^{D,M}_{m_2}\chi^{D,M}_{m_3}}.
\end{align}
\begin{enumerate}
    \item Dirichlet coarse graining, Dirichlet couplings and Dirichlet derivatives.
    
    This case was studied in \cite{RenormalisationIX}. The authors found that the couplings were already at the fixed point. This nice behaviour is due to the fact that, in this scenario, all the maps are cylindrically consistent, as depicted in \cref{cylindrical,commutation,demb}
    \item Dirichlet coarse graining, Haar couplings and Haar derivatives

    This case was also analysed \cite{RenormalisationIX} and led to non-trivial renormalisation flows, where the number of steps to reach the fixed point is parametrised by logarithmic-type functions for both $N_{M,b}$ and $N_M$ where $N$ was assumed to have compact momentum support.

    \item Dirichlet coarse graining, Haar couplings and Dirichlet derivatives.

    We analyse this further mixing in what follows. The Haar coupling $N_M$ does not depend on the derivative, hence its flow is equivalent to the one computed in \cite{RenormalisationIX}
    and flows to the fixed point as in 2. Therefore, we only need to renormalise $N_{M,b}$, where the mix between the Haar and Dirichlet happens. The renormalisation flow equation is given by
    \begin{align}
\label{flowfock}
N_M^{(r+1)}(m_1,m_2,m_3)=\frac{1}{M^9}\sum_{m'_1,m'_2,m'_3\in\mathbb{N}_{3M}^3}\prod_{s=1}^3i_{M,3M}(m'_s,m_s)\;N_{3M,b}^{(r)}(m_1',m_2',m_3'),
\end{align}
where we have chosen $M<M'(M):=3M$, $i_{M,3M}(\cdot,\cdot)$ is the kernel of $i_{M,3M}$ and $(r)$ labels the renormalisation step. Notice that the equation is independent of the chosen kernel. We start the flow by providing the initial Haar couplings $N^{(0)}_{M,b}$ with Dirichlet derivatives. We plug them into the renormalisation flow equations \cref{flowfock} to relate them to the --next--step--coupling
$N^{(1)}_{M,b}$
    \begin{align}
    \label{flowone}
    N_{M,b}^{(1)}(m_1,m_2,m_3)&=\frac{1}{9}\sum_{m_1',m'_2,m'_3\in \mathbb{N}_{3M}^3}\prod_{s=1}^3 I^D_{M3M}(m_s',m_s)N^{(0)}(\frac{m_1'}{3M})(\partial^D_{3M}\delta_{m_1'})(m_2')\delta_{m_1',m_3'}\nonumber \\ 
    &= \frac{1}{(3M)^3}\sum_{\substack{m_1'\in \mathbb{N}_{3M}^3\\ n_1,n_2,n_3 \in \mathbb{Z}_M^3}}e^{3M}_{n_1+n_3}(m_1')\prod_{s=1}^3e^M_{-n_s}(m_s)[\partial^D_{3M,b}\;e^{3M}_{n_2}(m_1')]\;N^{(0)}(\frac{m_1'}{3M})\nonumber \\
    &= \frac{1}{(3M)^3}\sum_{\substack{m_1'\in \mathbb{N}_{3M}^3\\ n_1,n_2,n_3 \in \mathbb{Z}_M^3}}(2\pi in_2^b)\;e^{3M}_{n_1+n_2+n_3}(m_1')\;N^{(0)}(\frac{m_1'}{3M})\prod_{s=1}^3e^{M}_{-n_s}(m_s) \nonumber \\
    &= \sum_{\substack{n_1,n_2,n_3 \in \mathbb{Z}_M^3\\n_o\in \mathbb{Z}_{M_0}^3}}(2\pi in_2^b)\;\delta_{n_1+n_2+n_3+n_0,0 \;\text{mod}(3M)}\hat{N}^{(0)}(n_0)\prod_{s=1}^3e^{M}_{-n_s}(m_s),
\end{align}
where in the second line we summed by parts and performed the $m_2',\;m_3'$ sums. 
In the third we applied the derivative and in the last line one we assumed that the lapse function has compact momentum support. 
This expression must flow to the fixed point value
\begin{align}
    N^{(*)}_{M,b}(m_1,m_2,m_3)=\sum_{\substack{n_1,n_2,n_3 \in \mathbb{Z}_M^3\\n_o\in \mathbb{Z}_{M_0}^3}}(2\pi in_2^b)\;\delta_{n_1+n_2+n_3+n_0,0}\;\hat{N}^{(0)}(n_0)\prod_{s=1}^3e^{M}_{-n_s}(m_s).
\end{align}
Expression (\ref{flowone}) differs from the one obtained for $N^{(1)}_{M}$ only by the factor $2\pi in^b_2$. Since the convergence to the fixed point is parametrised by the number of renormalisation steps $(r)$, its convergence rate is the same.
\item Dirichlet coarse graining, Dirichlet couplings and Haar derivative.

This case can be inferred from the previous cases. Since the coarse graining and the couplings are Dirichlet, the coupling $N_M$ is already at the fixed point. $N_{M,b}$ must be renormalised, its computation for the flow follows the same steps as  $N_{M,b}$ in the case 2.

 \end{enumerate}
\section{Functorial relations between \textbf{Seq} and \textbf{Func}}
The use of functors provides a natural and mathematically coherent way to relate the discrete and continuum  formulations of the Hamiltonian renormalisation programme. 
In our setting, the embedding and its adjoint act as functors between the categories \textsf{Seq} and \textsf{Func}.
This functorial viewpoint provides a unified categorical language in which lattice and functional renormalisation can be directly compared and related. 
Moreover, natural transformations between these functors encode physical equivalences between renormalisation flows, allowing one to identify when different discretisation or embedding choices lead to the same continuum theory. 
\begin{theorem}
    Let $\textbf{Seq}_D:=\textbf{Seq}\big|_{X=D}$ and $\textbf{Func}_D:=\textbf{Func}\big|_{X=D}$ be the Dirichlet subcategories of \textbf{Seq} and \textbf{Func} respectively. Then, the Dirichlet embedding $I^D$ is a functor from $\textbf{Seq}_D$ to $\textbf{Func}_D$, where its action on the objects $I_M^D:l_M\to L_M^D$  follows \cref{demb}, and its action on the morphisms closes the commutative diagrams  
    \[
\begin{array}{ccc}
\begin{tikzcd}
l_M \arrow[r, "i_{MM'}"] \arrow[d, "I_M^D"'] & l_{M'} \arrow[d, "I_{M'}^D"] \\
L_M^D \arrow[r, "I^D_{M,M'}"] & L_{M'}^D
\end{tikzcd}\;\;
&
\begin{tikzcd}
l_M \arrow[r, "\partial_M^D"] \arrow[d, "I_M^D"'] & l_M \arrow[d, "I_M^D"] \\
L_M^D \arrow[r, "\partial^D"] & L_M^D
\end{tikzcd}\;\;
&
\begin{tikzcd}
l_M \arrow[r, "\mathrm{Id}_{l_M}^D"] \arrow[d, "I_M^D"'] & l_M \arrow[d, "I_M^D"] \\
L_M^D \arrow[r, "\mathrm{Id}_{L_M}^D"] & L_M^D
\end{tikzcd}.
\end{array}
\]
Moreover, the adjoint map $\overline I_M^D$, is a functor from $\textbf{Func}_D$ to $\textbf{Seq}_D$ whose action on the objects $\overline{I}_M^D:L^D_M\to l_M$  follows \cref{demb}, and its action on the morphisms closes the aforementioned commutative diagrams with the appropriate inversion of arrows.  
\end{theorem}

\begin{proof}
We start by proving that the embedding is a functor for the objects. By definition $I^D_Mf_M=:F_M \in L^D_M$, for all $f_M\in l_M$, hence $I^D_M(l_M)=L^D_M$. For the morphisms in the Dirichlet subcategory, the commutative diagrams imply the relations
\begin{align}
    I_{M'}^D\circ i_{MM'} = I^D_{M,M'}\circ I_M^D,\qquad
I_M^D\circ\partial_M^D = \partial^D\circ I_M^D,\qquad
I_M^D\circ\mathrm{Id}_{l_M}^D = \mathrm{Id}_{L_M^D}\circ I_M^D.
\end{align}
To prove the closeness under composition, we glue two commutative diagrams at different resolutions $M<M'<M'' \in P_{\mathbb{O}}$
\[
\begin{tikzcd}[column sep=huge]
l_M \arrow[r, "i_{M,M'}"] \arrow[d, "I_M^D"'] &
l_{M'} \arrow[r, "i_{M',M''}"] \arrow[d, "I_{M'}^D"'] &
l_{M''} \arrow[d, "I_{M''}^D"] \\
L_M^D \arrow[r, "I^D_{M,M'}"] &
L_{M'}^D \arrow[r, "I_{M',M''}"] &
L_{M''}^D
\end{tikzcd}
\]
from which follows
\begin{align}
I^D(i_{M',M''}\;i_{M,M'})I^D(f_{M})&=I^D_{M''}(i_{M',M''}\;i_{M,M'}f_{M})=I^D_{M',M''}\;I^D_{M,M'}\;I^D_M(f_M) = I^D(i_{M',M''})\;I^D(i_{M,M'})\;I^D(f_M).
\end{align} 
For the adjoint, we notice that the composition diagram can be obtained by simply gluing the previous diagram vertically and changing the embedding labels 
\[
\begin{tikzcd}[row sep=huge, column sep=huge]
l_M \arrow[d, "I_M^D"'] \arrow[r, "i_{M,M'}"] & l_{M'} \arrow[d, "I_{M'}^D"'] \arrow[r, "i_{M',M''}"] & l_{M''} \arrow[d, "I_{M''}^D"'] \\
L_M^D \arrow[d, "\overline{I}_M^D"'] \arrow[r, "I^D_{M,M'}"] & L_{M'}^D \arrow[d, "\overline{I}^D_{M'}"'] \arrow[r, "I^D_{M',M''}"] & L_{M''}^D \arrow[d, "\overline{I}^D_{M''}"'] \\
l_M \arrow[r, "i_{M,M'}"] & l_{M'} \arrow[r, "i_{M',M''}"] & l_{M''}
\end{tikzcd}
\]
From the diagram it follows
\begin{align}
\overline{I}^D(I^D_{M',M''}\;I^D_{M,M'})\overline{I}^D(F_{M})&=\overline{I}^D_{M''}(I^D_{M',M''}\;I^D_{M,M'}F_{M})=i_{M',M''}\;i_{M,M'}\;\overline{I}^D_M(F_M)\nonumber \\
& = \overline{I}^D(I^D_{M',M''})\;\overline{I}^D(I^D_{M,M'})\;\overline{I}^D(F_M)
\end{align}
The prove of the closeness under composition of the other morphisms follows the same steps. 
\end{proof}
This result confirms that both the embedding and its adjoint preserve the categorical structure of the Dirichlet sector. The Haar embeddings are not functors of the Haar sector, this is again due to the lack of a continuum bounded operator $\partial^H$.

\section{Conclusion and outlook}
In this contribution, we have laid the ground for a categorical and functorial formulation of the Hamiltonian renormalisation programme 
providing an elegant bridge between functional and lattice renormalisation.
Furthermore, within this framework, we analysed Haar and Dirichlet embeddings and showed that the latter ---as well as its adjoint--- define functors between the corresponding subcategories $\textbf{Seq}_D$ and $\textbf{Func}_D$.

We have also revisited and expanded the analysis of the convergence rate of the renormalisation flow for the coupling parameters of the $U(1)^3$ model in $3+1$ Euclidean gravity, 
investigating the effects of different combinations of Dirichlet and Haar embeddings. Overall, our results highlight the utility of the categorical/functorial perspective as it  
allows a systematic classification of the required linear maps and helps to  
clarify equivalences and ambiguities in different discretisation/projection choices.

This work provides several open lines of research. For instance, in this contribution we only discussed Dirichlet and Haar embeddings but certainly this could be generalised to other interesting kernels, in particular, the generalisation of including Haar coarse grainings $i^H_M$ seems to follow directly from the theorems of section 3 with no further alteration. 

Using the Haar embeddings as a functor to access the Haar sector is however less straightforward since $\partial^H\notin$ \textbf{Func}. One may relax the definition of \textbf{Func} to enlarge the target space of $\partial^H$ to
include distributional derivatives. This comes with several challenges: $\partial^H$ is an unbounded operator in $L$, hence proofs relying on norm estimates or continuity in the $L$ must be adapted, furthermore, while $L_M^H$ is finite-dimensional, $\partial^H$ interpreted distributionally maps into the dual space $(L_M^H)'$, so derivatives are no longer strictly endomorphisms on objects. The equivalence between the two categories would still be somehow loose, since the functor from $\mathbf{Seq}$ to $\tilde{\mathbf{Func}}$ now maps some sequence morphisms into dual spaces. 

This contribution analysed different discretisation/projection choices in a categorical framework, however in order to include the convergence to the fix point one needs to include the notion of inductive limits. As mentioned in section 3, cylindrical consistency provides us with a direct system of Hilbert subspaces $L_M$ from which one can construct the direct limit. 

So far two functors within the Dirichlet subcategories have been identified, which seem to be adjoint functors, this however needs to be proven rigorously. Overall, in a  general sense, functors may  serve as a tool to categorise and compare different projection and embedding choices or analysing equivalences, and ambiguities in the projection/discretisation choices.

\printbibliography
\end{document}